\documentclass{eptcs}

\usepackage{times}
\usepackage{amsmath}
\usepackage{amssymb}
\usepackage{latexsym}
\usepackage{framed}
\usepackage[all,pdf]{xy}
\usepackage{color}
\usepackage{enumitem}
\usepackage{xspace}
\DeclareMathAlphabet{\mathcal}{OMS}{cmsy}{m}{n}	

\newcommand{\State}{\bullet}

\newcommand{\langdc}{\elcd}

\newcommand{\globallangc}{\rcd}
\newcommand{\rg}[1]{R_{#1}} 

\newcounter{questioncounter}

\renewcommand{\phi}{\varphi}

\newcommand{\set}[2]{\ensuremath{\{#1\,|\,#2\}}}

\newcommand{\ab}[1]{\langle#1\rangle}
\newcommand{\ag}{\textsc{ag}}
\newcommand{\at}{at}
\newcommand{\bis}{\rightleftarrows}

\newcommand{\CK}{\textbf{CK}}
\newcommand{\DK}{\textbf{DK}}
\newcommand{\elcd}{\ensuremath{\mathcal{ELCD}}\xspace}
\newcommand{\eld}{\ensuremath{\mathcal{ELD}}\xspace}
\newcommand{\gr}{\textsc{gr}}
\newcommand{\grouppath}{\ensuremath{G_1,\ldots,G_n}\xspace}
\renewcommand{\iff}{\text{iff}}
\newcommand{\Lra}{\Leftrightarrow}
\newcommand{\lra}{\leftrightarrow}

\newcommand{\mcs}{\ensuremath{\mathcal{S}}}
\newcommand{\mfm}{\mathfrak{M}}
\newcommand{\mfn}{\mathfrak{N}}
\newcommand{\mfp}{\mathfrak{P}}
\newcommand{\mfs}{\mathfrak{S}}
\newcommand{\modelsp}{\models_{\sf{p}}}
\newcommand{\MUpath}{|_{G_1}|\cdots|_{G_n}}
\newcommand{\MUpaththree}{|_{G_1}|_{G_2}|_{G_3}}
\newcommand{\pacd}{\ensuremath{\mathcal{PACD}}\xspace}
\newcommand{\premodelEx}{\ensuremath{(S,\backsim,V)}}
\newcommand{\prop}{\textsc{prop}}
\newcommand{\Ra}{\Rightarrow}
\newcommand{\ra}{\rightarrow}
\newcommand{\RCD}{\textbf{RCD}\xspace}
\newcommand{\RD}{\textbf{RD}\xspace}
\newcommand{\RDC}{\textbf{RCD}\xspace}
\newcommand{\PACD}{\textbf{PACD}\xspace}
\newcommand{\rcd}{\ensuremath{\mathcal{RCD}}\xspace}
\newcommand{\rd}{\ensuremath{\mathcal{RD}}\xspace}
\newcommand{\RELpath}{\mathrel{|_{G_1}|\cdots|_{G_n}}}
\newcommand{\RELpaththree}{\mathrel{|_{G_1}|_{G_2}|_{G_3}}}
\newcommand{\rest}{\mathrel{|}}
\newcommand{\RGpath}{R_{G_1}\cdots R_{G_n}}
\newcommand{\RR}{\textbf{RR}}
\newcommand{\ruleofC}{\text{RR}\ensuremath{_C}}
\newcommand{\SFIVE}{\textbf{S5}}
\newcommand{\SFIVED}{\textbf{S5D}}
\newcommand{\tbis}{\rightleftarrows^{\textsc{\scriptsize t}}}
\newcommand{\zigag}{zig\ensuremath{_{ag}}}

\newcommand{\ziggr}{zig\ensuremath{_{gr}}}

\usepackage[standard]{ntheorem}

\usepackage{enumitem}
\setenumerate[1]{itemsep=.2ex,partopsep=0pt,parsep=\parskip,topsep=.5ex}
\setitemize[1]{itemsep=.2ex,partopsep=0pt,parsep=\parskip,topsep=.5ex}
\setdescription{itemsep=.2ex,partopsep=0pt,parsep=\parskip,topsep=.5ex}

\title{Resolving Distributed Knowledge}

\begin{document}

\title{Resolving Distributed Knowledge}

\author{
Thomas {\AA}gotnes
\institute{University of Bergen, Norway}
\email{thomas.agotnes@uib.no}
\and
Y\`i N. W\'ang
\institute{Zhejiang University, China}
\email{ynw@zju.edu.cn}
}
\def\titlerunning{Resolving Distributed Knowledge}
\def\authorrunning{Thomas {\AA}gotnes \& Y\`i N. W\'ang}
\maketitle

\begin{abstract}
  \emph{Distributed knowledge} is the sum of the knowledge in a group;
  what someone who is able to discern between two possible worlds
  whenever \emph{any} member of the group can discern between them,
  would know. Sometimes distributed knowledge is referred to as the
  potential knowledge of a group, or the joint knowledge they could
  obtain if they had unlimited means of communication.  In epistemic
  logic, the formula $D_G\phi$ is intended to express the fact that
  group $G$ has distributed knowledge of $\phi$, that there is enough
  information in the group to infer $\phi$.  But this is not the same
  as reasoning about \emph{what happens if the members of the group
    share their information}. In this paper we introduce an operator
  $R_G$, such that $R_G \phi$ means that $\phi$ is true after $G$ have
  shared all their information with each other -- after $G$'s
  distributed knowledge has been \emph{resolved}.  The $R_G$
  operators are called \emph{resolution operators}.  Semantically, we
  say that an expression $R_G\phi$ is true iff $\phi$ is true in what
  van Benthem \cite[p. 249]{van2011logical} calls ($G$'s)
  \emph{communication core}; the model update obtained by removing
  links to states for members of $G$ that are not linked by \emph{all}
  members of $G$.  We study logics with different combinations of
  resolution operators and operators for common and distributed
  knowledge. Of particular interest is the relationship between
  distributed and common knowledge. The main results are
  sound and complete axiomatizations.
\end{abstract}

\section{Introduction}
\label{sec:introduction}

In epistemic logic \cite{fhmv95rk,meyer1995elai,vanDitmarsch07del}
different notions of \emph{group knowledge} describe different ways
in which knowledge can be associated with a group. \emph{Common
  knowledge} is stronger than individual knowledge: that something is
common knowledge requires not only that everybody in the group knows
it, but that everybody knows that everybody knows it, and so
on. \emph{Distributed knowledge}, on the other hand, is weaker than
individual knowledge: distributed knowledge is knowledge that is
distributed throughout the group even if no individual knows it.

More concrete informal descriptions of the concept of distributed
knowledge abound, but they are often inaccurate descriptions of the
concept as formalized in standard epistemic logic.  A misconception is that
something is distributed knowledge in a group if the agents in the
group could get to know it after some (perhaps unlimited)
communications between them\footnote{Some examples of informal
  descriptions of distributed knowledge from the literature include
  ``A group has distributed knowledge of a fact $\phi$ if the
  knowledge of $\phi$ is distributed among its members, so that by
  pooling their knowledge together the members of the group can deduce
  $\phi$'' \cite{fhmv95rk}; ``.. it should be possible for the members
  of the group to establish $\phi$ through communication''
  \cite{van1999group,roelofsen06dk}; ``.. the knowledge that would
  result of the agents could somehow 'combine' their knowledge''
  \cite{van1999group}. These descriptions can at least give a reader
  the impression that distributed knowledge is about internal
  communication in the group of agents.}. To see that this
interpretation must be incorrect, consider the formula $D_{\{1,2\}} (p
\wedge \neg K_1 p)$. In this formula, $D_G\phi$ and $K_i\phi$ mean
that $\phi$ is distributed knowledge in the group $G$, and individual
knowledge of agent $i$, respectively. Thus, the formula says that it
is distributed knowledge among agents $1$ and $2$ that $p$ is true and
that agent $1$ does not know $p$. This formula is \emph{consistent}
(also when we assume that knowledge has the S5 properties).  However,
it is not possible that agents $1$ and $2$ both can get to know that
$p$ is true and that agent $1$ does not know that $p$ is true
(assuming the S5 properties of knowledge), no matter how much they
communicate (or ``pool'' their knowledge).  The ``problem'' here is that in a formula $D_G\psi$,
$\psi$ describes the possible states of the world as they were before
any communication or other events took place, so a more accurate
reading of $D_{\{1,2\}} (p \wedge \neg K_1 p)$ would perhaps be that
it follows from the combination of $1$ and $2$'s knowledge that
\emph{$p \wedge \neg K_1p$ were true before any communication or other
  events took place}.  More technically, the ``problem'' is due to
the standard compositional semantics of modal logic: in the evaluation
of $D_G \phi$, the $D_G$ operator picks out a number of states
considered possible by the group $G$ (actually the states considered
possible by \emph{all} members of the group), and then $\phi$ is
evaluated in each of these states \emph{in the original model, without
  any effect of the $D_G$ operator}.

But we don't really consider this a problem. There are other
interpretations of distributed knowledge where the consistency of the
mentioned formula makes perfect sense, such that distributed knowledge
being the knowledge of a third party, someone ``outside the system''
who somehow has access to the epistemic states of the group members.
It shows, however, that it does not make sense to interpret
distributed knowledge as something that is true after the agents in
the group have communicated \emph{with each other} -- with the
standard semantics.

In this paper we introduce and study an alternative group modality
$R_G$, where $R_G \phi$ means (roughly speaking) that $\phi$ is true
after the agents in the group have shared all their information with
each other. We call that \emph{resolving} distributed knowledge, and
the $R_G$ operators are called \emph{resolution operators}.  

Semantically, we say that an expression $R_G\phi$ is true iff $\phi$
is true in what van Benthem \cite[p. 249]{van2011logical} calls
($G$'s) \emph{communication core}; the model update obtained by
removing links to states for members of $G$ that are not linked by
\emph{all} members of $G$.  See Fig. \ref{fig:example} for an
illustration.

\begin{figure}[h]
\begin{center}
\mbox{
 \xymatrix{
   \State^{\neg p}_s \ar@{..}^{1}[r]& \State_{t}^{p}  \ar@{..}^{2}[r]\ar@{..}^{1,2}[d]& \State_{u}^{\neg p}\\ 
   & \State^{p}_v \ar@{..}^{1}[d]\\
   & \State^{p}_w
}}
\mbox{
 \xymatrix{
   &\State_{t}^{p}  \ar@{..}^{1,2}[d]\\
   & \State^{p}_v 
}}
\end{center} 
\caption{Example taken from \cite[p. 248]{van2011logical}. Model on
  the left, its communication core (for the set of all agents
  $\{1,2\}$) on the right. Reflexivity, symmetry and transitivity are
  implicitly assumed.}
  \label{fig:example}
\end{figure}
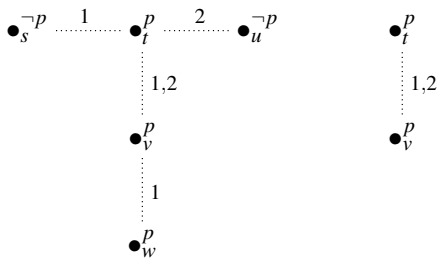
In this paper we capture that model transformation by the new
resolution operators, and study resulting logics. For example, the
formula $R_{\{1,2\}} (p \wedge \neg K_1 p)$ will be inconsistent in
the resulting logics. $R_{\{1,2\}} (p \wedge K_1 p)$ is true in state
$t$ in the model in Fig. \ref{fig:example}.  

This model transformation abstracts away from the issue of \emph{how}
the agents share their information; whether they \emph{communicate}
directly with each other and if so in which language, whether they are
informed by some outsider about the information other agents have and
if so how, and so on.  As noted by van Benthem
\cite[p. 249]{van2011logical}, the communication core cannot always be
obtained by \emph{public announcements} using the epistemic
language. Similarly, as noted by several researchers
\cite{van1999group,gerbrandy:1999,roelofsen06dk}, standard distributed
knowledge does not always follow logically from the knowledge of the
individual agents expressible in the epistemic language.  Our model,
like that of standard distributed knowledge, is purely semantic: we
assume that if an agent can discern between two different worlds, then
there exists some mechanism that results in other members of the group
being able to make the same distinction. This is further discussed
Section \ref{sec:conclusions}.

This model transformation models a particular kind of internal group
information sharing event. Exactly \emph{which} kind depends on what we
assume about what \emph{other} agents, i.e., agents that are not in
the group $G$ that resolve their knowledge, know about the fact that this
event is taking place. In this paper we will assume that it is common knowledge among the other
agents that $G$ resolve their knowledge -- but not what the
agents in $G$ actually learn. This corresponds to a natural class of
events: publicly observable private resolution of distributed
knowledge. An example is a meeting in a closed room, where it is
observed that a certain group meets in the room to share information.

We want to make it clear that we do not consider distributed knowledge
with standard semantics to be ``wrong''; the important thing is to be
clear about its meaning. In particular, the resolution operators are
not intended as a ``replacement'' of distributed knowledge operators,
but as a complement: they express different things. The logics we
study contains both types of operators, as well as common knowledge.
The main results are sound and complete axiomatizations.

Technically, the model transformation, which amounts to removing
certain edges, is similar to those found in the simplest dynamic
epistemic logics \cite{vanDitmarsch07del} such as public announcement
logics \cite{plaza1989logics}. \cite{van2011logical} has also pointed
out the close connection between the communication core and sequences
of public announcements. Public announcement logics with distributed
knowledge have been studied recently \cite{WA2013pacd}.  In
the absence of \emph{common} knowledge, we get reduction axioms for
public announcement logic with distributed knowledge. This turns out
to be the case for resolution operators as well.  It is not the case
in the presence of common knowledge, however.  

There is a close connection between the communication core and
common knowledge \cite{van2011logical}. By studying complete
axiomatizations of logics with the resolution operators we make some
aspects of that connection precise and give an answer to the question
``when does distributed knowledge become common knowledge?'' -- under
certain assumptions.

The rest of the paper is organized as follows. In the next section we
review some background definitions and results from the literature,
before we introduce logics with the new resolution operators in Section
\ref{sec:resolving}, where we also look at some properties of the
operators. In Section
\ref{sec:completeness} we prove completeness of resulting logics; the
most interesting case being epistemic logic with common and
distributed knowledge and resolution operators.  We discuss related and
future work and conclude in Section \ref{sec:conclusions}.

\section{Background}
\label{sec:background}

In this section we give a (necessarily brief) review of the main
background concepts from the literature.

We henceforth assume a countable set of propositional variables
$\prop$ and a finite set of agents $\ag$. We let \gr\ be the set of
all non-empty groups, i.e., $\gr=\wp(\ag)\setminus\emptyset$.

An \emph{epistemic model} over \prop\ and \ag\ (or just a \emph{model})
$\mfm = (S,\sim, V)$ where $S$ is a set of states (or worlds), $V:
\prop \rightarrow 2^S$ associates a set of states $V(p)$ with each
propositional variable $p$, and $\sim$ is a function that maps each
agent to a binary equivalence relation on $S$. We write $\sim_i$ for
$\mathord{\sim}(i)$.

$s \sim_i t$ means that agent $i$ cannot discern between states $s$
and $t$ -- if we are in $s$ she doesn't know whether we are in $t$,
and vice versa. Considering the distributed knowledge of a group $G$
-- a key concept in the following -- we define a derived relation
$\sim_G = \bigcap_{a\in G}\sim_a$ (it is easy to see that $\sim_G$ is
an equivalence relation). Intuitively, someone who has all the
knowledge of all the members of $G$ can discern between two states if
and only if at least one member of $G$ can discern between them. We
will also consider common knowledge. A similar relation modeling the
common knowledge of a group is obtained by taking the transitive
closure of the union of the individual relations: $\backsim_{C_G} =
(\bigcup_{i\in G}\sim_i)^*$.

\begin{definition}
Below are several languages from the literature.
$$\begin{array}{ll}
(\eld)& \phi ::= p \mid \neg \phi \mid \phi \wedge \phi \mid K_i \phi \mid D_G \phi\\
(\elcd)&\phi ::= p \mid \neg \phi \mid \phi \wedge \phi \mid K_i \phi \mid D_G \phi \mid C_G \phi\\
(\pacd)&\phi ::= p \mid \neg \phi \mid \phi \wedge \phi \mid K_i \phi \mid D_G \phi \mid C_G \phi\mid [\phi]\phi,
\end{array}$$
where $p \in \prop$, $i \in \ag$ and $G \in\gr$. We use the usual
propositional derived operators, as well as $E_G\phi$ for $\bigwedge_{i\in G}K_i\phi$.
\end{definition}
\eld and \elcd are static epistemic languages with distributed
knowledge, and with distributed and common knowledge,
respectively. These are the languages we will extend with resolution
operators in the next section. We will also be interested in \pacd,
the language for public announcement logic with both
common knowledge and distributed knowledge, when we look at
completeness proofs.

Satisfaction of a formula $\varphi$ of any of these languages in a state
$m$ of a model $\mfm$, denoted $\mfm,m \models \varphi$, is defined
recursively by the following clauses:
$$\begin{array}{@{}ll@{\ \ \iff\ \ }l@{}}
&\mfm,m\models p 				&m\in V(p)\\
&\mfm,m\models\neg\varphi		&\mfm,m\not\models\varphi\\
&\mfm,m\models\varphi\wedge\psi	&\mfm,m\models\varphi\ \&\ \mfm,m\models\psi\\
&\mfm,m\models K_a \phi& \forall n\in S.\ (m\sim_a n\Ra \mfm,n\models\phi)\\
&\mfm,m\models D_G\phi		&\forall n\in S.\ (m \sim_G n\Ra\mfm,n\models\phi)\\
&\mfm,m\models C_G\phi		&\forall n\in S.\ (m (\bigcup_{i \in G}\sim_i)^* n\Ra\mfm,n\models\phi)\\
&\mfm,m\models[\psi]\phi&\mfm,m\models\psi\Ra\mfm|\psi,m\models\phi.
\end{array}$$
where $R^*$ denotes the transitive closure of $R$ and $\mfm|\psi$ is the submodel
of $\mfm$ restricted to $\set{m\in M}{\mfm,m\models\psi}$.
\emph{Validity} is defined as usual: $\models \phi$ means that $\mfm,m
\models \phi$ for all $\mfm$ and $m$.

We now define some axiom schemata and rules. The classical ``S5'' proof system for multi-agent epistemic logic, denoted (\SFIVE), consists of the following axioms and rules:
$$\begin{tabular}{ll}
(PC)& 	instances of tautologies\\
(K)& 	$K_i(\phi\rightarrow\psi)\rightarrow(K_i\phi\rightarrow K_i\psi)$\\
(T)&	$K_i\phi\ra\phi$\\
(4)&	$K_i\phi\ra K_iK_i\phi$\\
(5)&	$\neg K_i\phi\ra K_i\neg K_i\phi$\\
(MP)& 	from $\phi$ and $\phi\ra\psi$ infer $\psi$\\
(N)&	from $\phi$ infer $K_i\phi$.
\end{tabular}$$

Axioms for distributed knowledge, denoted (\DK):
$$\begin{tabular}{ll@{\quad}ll}
(K$_D$)&$D_G(\phi\ra\psi)\ra(D_G\phi\ra D_G\psi)$\\
(T$_D$)&$D_G\phi\ra\phi$\\
(5$_D$)&$\neg D_G\phi\ra D_G\neg D_G\phi$\\
(D1)&$K_i\phi\lra D_i\phi$\\
(D2)&$D_G\phi\ra D_H\phi$, if $G\subseteq H$.
\end{tabular}$$

Axioms and rules for common knowledge, denoted (\CK):
$$\begin{tabular}{ll@{\quad}ll}
(K$_C$)&$C_G(\phi\ra\psi)\ra(C_G\phi\ra C_G\psi)$\\
(T$_C$)&$C_G\phi\ra\phi$\\
(C1)&$C_G\phi\ra E_GC_G\phi$\\
(C2)&$C_G(\phi\ra E_G\phi)\ra(\phi\ra C_G\phi)$\\
(N$_C$)& from $\phi$ infer $C_G\phi$.
\end{tabular}$$

The system that consists of (\SFIVE) and (\DK) over the language
$\eld$, denoted \SFIVED, is a sound and complete axiomatization of all $\eld$
validities.  The system that consists of (\SFIVE), (\DK) and (\CK)
over the language $\elcd$ is a sound and complete axiomatization of
all $\elcd$ validities.

\section{\bf\fontsize{11pt}{1em}\selectfont Resolving Distributed Knowledge}
\label{sec:resolving}

We want to model the event that $G$ resolves their knowledge. An
immediate question is: whenever the group $G$ is a proper subset of
the set of all agents, what do \emph{the other} agents know about the
fact that this event takes place? Here we will model situations where
it is common knowledge among the other agents \emph{that} the event
takes place, but not \emph{what} the members of the group learn.  As
discussed in the introduction, this corresponds to a natural class of
information sharing events, namely publicly observable private
communication, such as a meeting in a closed room that is observed to
be taking place. This is captured by a \emph{global} model update: in
every state, remove a link to another state for any member of $G$
whenever it is not the case that there is a link to that state for
\emph{all} members of $G$.

Formally, given a model $\mfm = (S,\sim, V)$ and a group of agents
$G$, the \emph{(global) $G$-resolved update of $\mfm$} is the model
$\mfm|_G$ where $\mfm|_G =
(S,\sim\rest_G, V)$ and
$$(\sim\rest_G)_i = \left\{\begin{array}{ll}
        \bigcap_{j \in G} \sim_j, & i \in G,\\
        \sim_i, & \mbox{otherwise}.
      \end{array}\right.$$

We consider the following new languages with resolution operators.
\begin{definition}[Languages]
$$\begin{array}{ll}
(\rd)& \phi ::= p \mid \neg \phi \mid \phi \wedge \phi \mid K_i \phi \mid D_G \phi \mid \rg{G} \phi\\
(\rcd)&\phi ::= p \mid \neg \phi \mid \phi \wedge \phi \mid K_i \phi \mid D_G \phi \mid C_G \phi \mid \rg{G}\phi,\\
\end{array}$$
where $p \in \prop$, $i \in \ag$ and $G \in\gr$.
\end{definition}

The interpretation of these languages in a pointed model is defined as usual, with the following additional clause for the resolution operator:
\[\mfm,s \models \rg{G} \phi \quad \text{iff} \quad \mfm|_G,s \models \phi.\]

A couple of observations. Recall that we write $\sim_H$ for $\bigcap_{i\in H}\sim_i$. Thus,
$$(\sim\rest_G)_i=\left\{\begin{array}{@{}l@{\ \ }l@{}}
\sim_G,&i\in G,\\
\sim_i,&i\notin G.\\
\end{array}\right.
\quad
(\sim\rest_G)_H=\left\{\begin{array}{@{}l@{\ \ }l@{}}
\sim_H,&G\cap H=\emptyset,\\
\sim_{G\cup H},&G\cap H\neq\emptyset.\\
\end{array}\right.$$
Also note that $(\sim\rest_G)_i=(\sim\rest_G)_{\{i\}}$.

\subsection{Some Validities}
\label{sec:reduction}

Let us start with a trivial validity: resolution has no effect for a
singleton coalition.
\begin{proposition}\label{prop:val1}
The following is valid, where $i\in\ag$ and $\phi\in\rcd$: $R_{\{i\}}\phi\lra\phi$.
\end{proposition}

More interesting are the following properties.
\begin{proposition}[Reduction Principles]\label{prop:reduction}
\ \  The following are valid, where $G,H\in\gr$, $p \in \prop$ and $\phi\in\rcd$:
  \begin{enumerate}
  \item\label{red-base} $\rg{G} p \leftrightarrow p$ 
  \item\label{red-conj} $\rg{G}(\phi \wedge \psi) \leftrightarrow \rg{G}\phi \wedge \rg{G}\psi$
  \item\label{red-neg} $\rg{G}\neg\phi \leftrightarrow \neg\rg{G}\phi$
  \item $\rg{G}K_i \phi \leftrightarrow D_{G} \rg{G}\phi$, when $i \in G$ 
  \item $\rg{G}K_i \phi \leftrightarrow K_i \rg{G}\phi$, when $i \not\in G$ 
  \item $\rg{G}D_H\phi \leftrightarrow D_{G \cup H}\rg{G}\phi$, when $G \cap H \neq \emptyset$ 
  \item $\rg{G}D_H\phi \leftrightarrow D_{H}\rg{G}\phi$, when $G \cap H = \emptyset$.
  \end{enumerate}
\end{proposition}

These properties are reduction principles, of the type known from
public announcement logic: they allow us to simplify expressions
involving resolution operators. If we have such principles for the
combination of resolution with all other operators we can eliminate
resolution operators altogether. There are two cases missing above:
$R_GC_H$ and $R_GR_H$\footnote{The lack of a reduction axiom for the
  general $\rg{G}\rg{H}\phi$ case does not mean we cannot get a
  reduction in the language \rd: we can simply do the reduction
  ``inside-out''.}. We consider them next.

\subsubsection{Common Knowledge}

First, after the \emph{grand} coalition have resolved their knowledge, then all the
distributed information in the system is common knowledge: there is no longer
a distinction between distributed and common knowledge:

\begin{proposition}\label{red-grand-coalition}
For any $\phi \in \rcd$: $R_{\ag}C_{\ag}\phi\lra R_{\ag}D_{\ag}\phi$.
\end{proposition}
\begin{proof}
\ref{red-grand-coalition}.
Given a model $\mfm=(S,\sim,V)$ and $s\in S$,
$$\begin{array}{llll}
&\mfm,s\models R_{\ag}C_{\ag}\phi\\
\iff&\mfm|_{\ag},s\models C_{\ag}\phi\\
\iff&\forall t\in S.\ (s(\sim\rest_{\ag})_{C_{\ag}}t\Ra\mfm|_{\ag},s\models\phi)\\
\iff&\forall t\in S.\ (s\sim_{\ag}t\Ra\mfm|_{\ag},s\models\phi)\quad(\dag)\\
\iff&\mfm|_{\ag},s\models D_{\ag}\phi\\
\iff&\mfm,s\models R_{\ag}D_{\ag}\phi,\\
\end{array}$$
where for $(\dag)$ we show that $(\sim\rest_{\ag})_{C_{\ag}}=\sim_{\ag}$. This is easy: by definition we can verify that for all $i\in\ag$, $(\sim\rest_{\ag})_i=\sim_{\ag}$; hence $(\sim\rest_{\ag})_{C_{\ag}}=(\bigcup_{i\in\ag}(\sim\rest_\ag)_i)^*=(\sim_\ag)^*=\sim_{\ag}$.
\end{proof}

For the general case, as in the case of distributed knowledge, we have that the resolution operators and common knowledge operators commute when the groups are \emph{disjoint}:

\begin{proposition}\label{prop:red-ck}
Let $i$ be an agent, $G$ and $H$ groups of agents and $\phi \in \rcd$. The following hold:
\begin{enumerate}
\item\label{RC-no-intersection}   If $G \cap H = \emptyset$, then $\models\rg{G}C_H \phi \leftrightarrow C_H\rg{G}\phi$
\item\label{RC-subset} If $G\supseteq H$ and $i\in G$, then $\models R_GC_H\phi\lra R_GK_i\phi\lra D_GR_G\phi$.
\end{enumerate}
\end{proposition}
\begin{proof}
See the appendix.
\end{proof}

However, this does not hold for overlapping groups $G$ and $H$. In
general, we have that (see the proof of the proposition above) $\mfm,s
\models \rg{G}C_H \phi$ iff $\mfm|_G,t \models \phi$
for any $(s,t) \in \sim^{*'}_H$, where $\sim^{*'}_H = (\bigcap_{i \in
  G} \sim_i \cup \bigcup_{i \in H \setminus G} \sim_i)^*$.  This does
not seem to be reducible.

\subsubsection{Iterated resolution}

What about $\rg{G}\rg{H}\phi$? In extreme cases, we have:
\begin{proposition}
The following are valid, where $G,H\in\gr$ and $\phi\in\rcd$:
\begin{enumerate}
\item $R_GR_H\phi\lra R_HR_G\phi$, if $G\cap H=\emptyset$
\item $R_GR_G\phi\lra R_G\phi$.
\end{enumerate}
\end{proposition}
However, in the general case there does not seem to be a reduction
axiom in this case. In particular, $R_GR_H\phi$ is not equivalent to
$\rg{G \cup H}\phi$.

Let us consider an example of iterated resolution.

\begin{example}[Triple update]\label{ex:triup}
Let $\mfm = (S, \sim, V)$ and\\ $\mfm\MUpaththree = (S, \sim\RELpaththree, V)$. For any agent $i$,
for any number $x$, we write $G_x$ for ``$i\in G_x$'', and $\overline{G_x}$ for ``$i\notin G_x$''. Then
$$(\sim\RELpaththree)_i=\left\{
\begin{array}{@{\text{if }}l@{\ \ }l}
\overline{G_1G_2G_3}:&\sim_i\\
G_1\overline{G_2G_3}:&\sim_{G_1}\\
\overline{G_1}G_2\overline{G_3}
	\left\{\begin{array}{@{}l@{}}
	G_1\cap G_2=\emptyset:\\
	G_1\cap G_2\neq\emptyset:\\
	\end{array}\right.&
	\begin{array}{@{}l@{}}
	\sim_{G_2}\\
	\sim_{G_1\cup G_2}\\
	\end{array}\\
G_1G_2\overline{G_3}:&\sim_{G_1\cup G_2}\\ 
\overline{G_1G_2}G_3:&\sim_{G_3}\\
G_1\overline{G_2}G_3:&\sim_{G_1\cup G_3}\\
\overline{G_1}G_2G_3
	\left\{\begin{array}{@{}l@{}}
	G_1\cap G_2=\emptyset:\\
	G_1\cap G_2\neq\emptyset:\\
	\end{array}\right.&
	\begin{array}{@{}l@{}}
	\sim_{G_2\cup G_3}\\
	\sim_{G_1\cup G_2\cup G_3}\\
	\end{array}\\
G_1G_2G_3:&\sim_{G_1\cup G_2\cup G_3}\\
\end{array}\right.$$
\end{example}

In general we get the following (the proof is straightforward from the
semantic definition).
\begin{proposition}\label{prop:res-red}
Let $M = (S, \sim, V)$ and $M\MUpath = (S, \sim\RELpath, V)$. Then, following the notation of Example \ref{ex:triup}, for any $i\in\ag$,
  $$(\sim\RELpath)_i=\left\{
  \begin{array}{ll}
  \sim_i,&\text{if }\overline{G_1\cdots G_n}\\
  \sim_{G_1\cup\Theta},&\text{if starting with }\overline{G_1}G_2\text{ and }\\
										 & G_1\cap G_2\neq\emptyset\\
  \sim_\Theta,&\text{otherwise}\\
  \end{array}\right.$$
where $\Theta$ is the union of all $G_x$ such that $i\in G_x$.
\end{proposition}

\subsection{Reduction Normal Form for Individual and Distributed Knowledge}
As we see from the previous section, the reduction axioms for
individual knowledge and distributed knowledge both contain two
distinct cases, and the principles of iterated resolution become more
complicated. In this section we give a unique form for such
reductions, which will be of use later when we prove completeness. We
shall call it \emph{reduction normal form} for individual and
distributed knowledge.

\begin{definition}[$\delta$ function]
Given an agent $i$, a group $H$, and a sequence of groups \grouppath, we define a function $\delta$ as follows:
$$\begin{array}{r@{\ }c@{\ }l}
\delta_0&=&\left\{\begin{array}{l@{\hspace{3.16em}}l}G_n\cup H,&G_n\cap H\neq\emptyset\\H,&G_n\cap H=\emptyset\end{array}\right.\\[2ex]
\delta_{x}&=&\left\{\begin{array}{ll}G_{n-x}\cup\delta_{x-1},&G_{n-x}\cap \delta_{x-1}\neq\emptyset\\\delta_{x-1},&G_{n-x}\cap\delta_{x-1}=\emptyset\end{array}\right.\\[2ex]
\delta(H,\grouppath)&=&\quad\ \delta_n.\\
\end{array}$$

Clearly
$\delta(H,\grouppath)\subseteq H\cup G_1\cup\cdots\cup G_n$. We simply write $\delta$ instead of $\delta(H,\grouppath)$ when its parameters are clear in the context.
\end{definition}

\begin{proposition}\label{cor-rd-red}
Let $i\in\ag$, $\grouppath,H\in\gr$, $M = (S, \sim, V)$ and
$M\MUpath = (S, \sim\RELpath, V)$. Then,
\begin{enumerate}
\item\label{item-k-red} $\models \RGpath K_i\phi\lra D_{\delta(\{i\},\grouppath)} \RGpath \phi$;
\item\label{item-d-red} $\models \RGpath D_H\phi\lra D_{\delta(H,\grouppath)} \RGpath \phi$;
\item\label{item-up-ag} $(\sim\RELpath)_i=\sim_{\delta(\{i\},\grouppath)}$;
\item\label{item-up-gr} $(\sim\RELpath)_H=\sim_{\delta(H,\grouppath)}$.
\end{enumerate}
\end{proposition}
\begin{proof}
Straightforward: the recursive steps in the definition of the $\delta$ function matches exactly the reduction axioms. Note that clauses \ref{item-k-red} and \ref{item-up-ag} can be treated as special cases of clauses \ref{item-d-red} and \ref{item-up-gr} respectively.
\end{proof}

\section{Axiomatizations}
\label{sec:completeness}

We construct sound and complete axiomatizations of the logics
for the two languages \rd and \rcd.

\subsection{Resolution and Distributed Knowledge}

Consider the language \rd. Let \RD\ be the system defined in Figure
\ref{fig:rd}, where (\SFIVE) and (\DK) are found in Section
\ref{sec:background} and (\RR) stands for the following reduction axioms for resolution:
$$\begin{tabular}{ll@{\quad}ll}
(RA)&$\rg{G} p \leftrightarrow p$ \\
(RC)&$\rg{G}(\phi \wedge \psi) \leftrightarrow \rg{G}\phi \wedge \rg{G}\psi$\\
(RN)&$\rg{G}\neg\phi \leftrightarrow \neg\rg{G}\phi$\\
(RD1)&$\rg{G}D_H\phi \leftrightarrow D_{G \cup H}\rg{G}\phi$, if $G \cap H \neq \emptyset$ \\
(RD2)&$\rg{G}D_H\phi \leftrightarrow D_{H}\rg{G}\phi$, if $G \cap H = \emptyset$.
\end{tabular}$$
Note that (\RR) contains most of the validities in Proposition
\ref{prop:reduction}, except for the reduction principles for
individual knowledge -- they are provable with RD1, RD2 and D1. In
addition, we need the rule N$_R$ for making a reduction to
\SFIVED. With the rule N$_R$ we can easily show that the rule of
\emph{Replacement of Equivalents (RoE)} is admissible in \RD. RoE
allows us to carry out a reduction even without having a reduction axiom for iterated resolution.

\begin{figure}[htbp]
\centering
\fbox{%
\begin{tabular}{@{\ }ll@{\ }}
(\SFIVE)& classical proof system for multi-agent epistemic logic\\
(\DK)	& characterization axioms for distributed knowledge\\
(\RR)	& reduction axioms for resolution\\
(N$_R$) & from $\phi$ infer $R_G\phi$\\
\end{tabular}
}
\caption{Axiomatization \RD.}
\label{fig:rd}
\end{figure}

\begin{theorem}
 Any \rd formula is valid if and only if it is provable in \RD.
\end{theorem}

\subsection{Resolution, Distributed and Common Knowledge}

Consider the language \rcd. Let \RCD\ be the system defined in Figure
\ref{fig:rcd}, which extends \RD\ with (\CK), found in Section
\ref{sec:background}, and an induction rule for resolved common
knowledge (\ruleofC).

\begin{figure}[htbp]
\centering
\fbox{
\begin{tabular}{@{}ll@{}}
(\SFIVE)& classical proof system for multi-agent epistemic logic\\
(\CK)	& axioms and rules for common knowledge\\
(\DK)	& characterization axioms for distributed knowledge\\
(N$_R$)&	from $\phi$ infer $R_G\phi$\\
(\RR)	& reduction axioms for resolution\\
(\ruleofC)& from $\phi\ra(E_H\phi\wedge\RGpath\psi)$ infer\\ & $\phi\ra\RGpath C_H\psi$
\end{tabular}
}
\caption{Axiomatization \RDC}
\label{fig:rcd}
\end{figure}

\subsubsection{Soundness}

For soundness it suffices to show that the rule \ruleofC\ preserves
validity (we know that the other axioms/rules are valid/validity
preserving from soundness results for the logics based on the sublanguages of \rcd).
\begin{lemma}[\ruleofC-validity preservation] For all \rcd\ formulas $\phi$ and $\psi$, all $\grouppath, H\in\gr$, 
if $\models\phi\ra(E_H\phi\wedge\RGpath\psi)$, then $\models\phi\ra\RGpath C_H\psi$.
\end{lemma}
\begin{proof}
Suppose $\models \phi\ra(E_H\phi\wedge\RGpath\psi)$. Given a model $\mathfrak{M}$ and a state $s$, suppose $\mathfrak{M},s\models\phi$, we must show that $\mathfrak{M},s\models \RGpath C_H\psi$, i.e., $\mathfrak{M}\MUpath,s\models C_H\psi$. Thus, for all $H$-paths $s_0(\sim\MUpath)_{i_0}\cdots(\sim\MUpath)_{i_{x-1}}s_{x}$, where $s=s_0$, we need to show that $\mathfrak{M}\MUpath,s_x\models\psi$.

From $\models \phi\ra(E_H\phi\wedge\RGpath\psi)$ and $\mfm,s_0\models\phi$ we get $\mathfrak{M},s_0\models(E_H\phi\wedge\RGpath\psi)$, which entails:
$$\mfm,s_1\models\phi\quad\text{and}\quad\mfm\MUpath,s_0\models\psi.$$
From $\mfm,s_1\models\phi$ we get $\mathfrak{M},s_1\models(E_H\phi\wedge\RGpath\psi)$, which entails:
$$\mfm,s_2\models\phi\quad\text{and}\quad\mfm\MUpath,s_1\models\psi.$$
By similar reasoning, for all $y=0,\ldots,x$, we have
$$\mfm,s_y\models\phi\quad\text{and}\quad\mfm\MUpath,s_y\models\psi,$$
which entails $\mfm\MUpath,s_x\models\psi$ as we wish to show.
\end{proof}

\begin{corollary}[Soundness]
  For any \rcd formula $\phi$, if $\phi$ is provable
  in \RCD, then it is valid.
\end{corollary}

\subsubsection{Completeness}

As already discussed, \RCD is similar to \PACD (axiomatization for public announcement
logic with common and distributed knowledge; see \cite{WA2013pacd}): both logics extend
epistemic logic with common and distributed knowledge with dynamic
operators with update semantics that remove states. There does not seem,
however, to be a trivial relationship between
the two types of dynamic operators. We are nevertheless able to make
heavy use of the completeness proof of \PACD\ in \cite{WA2013pacd} when
proving completeness of \RCD.  That proof is again based on the
completeness proof for public announcement logic with (only) common knowledge found in
\cite{baltag98logic,vanDitmarsch07del}, extended to deal with the
distributed knowledge operators (which is non-trivial since
intersection is not modally definable). In the following completeness
proof we tweak the \PACD\ proof to deal with resolution operators
instead of public announcement operators. The general proof strategy
is as follows: define a finite canonical pseudo model, where
distributed knowledge operators are taken as primitive, and then
transform it to a proper model while preserving truth. For the last
step we can use a transformation based on \emph{unraveling and
  folding} in \cite{WA2013pacd} directly.

The most important difference to the \PACD\ completeness proof in
\cite{WA2013pacd}, and indeed the crux of the proof, is the use of the
induction rule for resolved common knowledge (\ruleofC). No
corresponding rule is needed in the \PACD completeness proof. The rule
is used in the proof of Lemma \ref{lemma1}(\ref{itemBC}).

\paragraph{Pseudo Semantics}

\begin{definition}[Pre-models\cite{WA2013pacd}]\label{def-pre-model}
A \emph{pre-model} is a tuple $\mfm=\premodelEx$ where:
\begin{itemize}
\item $S$ is a non-empty set of states;
\item $\backsim$ is a function which maps every agent and every non-empty group of agents to an equivalence relation; we write $\backsim_i$ and $\backsim_G$ for $\mathord{\backsim}(i)$ and $\mathord{\backsim}(G)$ respectively;
\item $V: \prop\to\wp(S)$ is a valuation.
\end{itemize}

$\backsim_{C_G}$ is defined as the reflexive transitive closure of $\bigcup_{i\in G}\backsim_i$, just as for a model.
\end{definition}

A pre-model is technically a model with a bigger set of agents (all groups are treated as agents in a pre-model). More precisely, if we make the set of agents $A$ explicit in a pre-model, e.g., $\mfm=(A,S,\backsim,V)$, then $\mfm$ is in fact a ``genuine'' model $(S,\backsim,V)$ where the set of agents is $A\cup(\wp(A)\setminus\emptyset)$.

\begin{definition}[Pseudo models\cite{WA2013pacd}]\label{def-pseudo-model}
A \emph{pseudo model} is a pre-model $\mfm=\premodelEx$ such that for any agent $i$ and any groups $G$ and $H$,
\begin{itemize}
\item $\backsim_{\{i\}}=\backsim_i$, and
\item $G\subseteq H$ implies $\backsim_H\subseteq\backsim_G$.
\end{itemize}
A \emph{pointed pre-model} (resp. \emph{pointed pseudo model}) is a tuple $(\mfm,s)$ consisted of a pre-model (resp. pseudo model) $\mfm$ and a state $s$ in $\mfm$.
\end{definition}

\begin{definition}[Pseudo semantics]\label{def:ps-semantics}
Given a pre-model $\mfm=\premodelEx$, let $m$ be a state in $M$. \emph{Satisfaction at} $(\mfm,s)$ is defined as follows:
$$\begin{array}{lll}
\mfm,s\modelsp p 				&\iff&s\in V(p)\\
\mfm,s\modelsp\neg\phi		&\iff&\mfm,s\not\modelsp\phi\\
\mfm,s\modelsp\phi\wedge\psi	&\iff&\mfm,s\modelsp\phi\ \&\ \mfm,s\modelsp\psi\\
\mfm,s\modelsp K_i\phi		&\iff&(\forall t\in\mfm)(s\backsim_in\Ra\mfm,t\modelsp\phi)\\
\mfm,s\modelsp C_G\phi		&\iff&(\forall t\in\mfm)(s\backsim_{C_G}t\Ra\mfm,t\modelsp\phi)\\
\mfm,s\modelsp D_G\phi		&\iff&(\forall t\in\mfm)(s\backsim_Gt\Ra\mfm,t\modelsp\phi)\\
\mfm,s\modelsp R_G\psi		&\iff&\mfm|_G,s\modelsp\phi,\\
\end{array}$$
where $\mfm|_G=(S,\backsim\rest_G,V)$ such that
$$(\backsim\rest_G)_i=\left\{\begin{array}{@{}l@{\ }l@{}}
\backsim_G,&i\in G\\
\backsim_i,&i\notin G\\
\end{array}\right.
\text{\ \ and\ \ }
(\backsim\rest_G)_H=\left\{\begin{array}{@{}l@{\ }l@{}}
\backsim_{H\cup G},&H\cap G\neq\emptyset\\
\backsim_H,&H\cap G=\emptyset\\
\end{array}\right.
$$

\emph{Satisfaction in a pre-model} $\mfm$ (denoted by $\mfm\modelsp\phi$) is defined as usual. We use $\modelsp\phi$ to denote validity, i.e. $\mfm,s\modelsp\phi$ for any pointed pre-model $(\mfm,s)$. We write $\models$ instead of $\modelsp$ when there is no confusion.
\end{definition}

\begin{proposition}\label{prop:up-pres-ps}
Let $\mfm$ be a pseudo model, $G$ a group of agents. Then $\mfm|_G$ is a pseudo model.
\end{proposition}
\begin{proof}
See the appendix.
\end{proof}

\begin{proposition}\label{prop:res-red-pseudo}
Propositions \ref{prop:res-red} and \ref{cor-rd-red} still hold for pseudo models.
\end{proposition}

When we regard a pre-model as a genuine model, classical (individual)
bisimulation becomes an invariance relation. To make this clear, we first elaborate the definition of bisimulation for pre-models, and then introduce its invariance results.

\begin{definition}[Pre-model bisimulation]\label{def:standard-bis}
Let two pre-models $\mfm=(S,\backsim,V)$ and $\mfm'=(S',\backsim',V')$ be given. A non-empty relation $Z\subseteq S\times S'$ is called a \emph{bisimulation} between $\mfm$ and $\mfm'$, denoted by $\mfm\bis\mfm'$, if for all $\tau\in\ag\cup\gr$, all $s\in S$ and $s'\in S'$ such that $sZs'$, the following hold.
\begin{description}
 \item[(\at)] For all $p\in\prop$, $s\in V(p)$ iff $s'\in V'(p)$;
 \item[(zig)] For all $t\in S$, if $s\sim_\tau t$, there is a $t'\in S'$ such that $s'\sim'_\tau t'$ and $tZt'$;
 \item[(zag)] For all $t'\in S'$, if $s'\sim'_\tau t'$, there is a $t\in S$ such that $s\sim_\tau t$ and $tZt'$.
 \end{description}
We say that pointed pre-models $(\mfm,s)$ and $(\mfm',s')$ are
\emph{bisimilar}, denoted $(\mfm,s)\bis(\mfm',s')$, if there is a bisimulation $Z$ between $\mfm$ and $\mfm'$ linking $s$ and $s'$.
\end{definition}

\begin{proposition}\label{prop:bis-inv}
\hspace{-1ex}Resolution preserves pre-model bisimulation. I.e., for all pointed pre-models $(\mfm,s)$ and $(\mfm',s')$, if $(\mfm,s)\bis(\mfm',s')$ then $(\mfm|_G,s)\bis(\mfm'|_G,s')$.
\end{proposition}
\begin{proof}
See the appendix.
\end{proof}

\begin{corollary}
  \label{cor:premodelinvariance}
  For any pre-models $\mfm$ and $\mfm'$, if $(\mfm,s) \bis (\mfm',s')$
  then $\mfm,s \modelsp \phi$ iff $\mfm',s' \modelsp \phi$ for any \rcd
  formula $\phi$.
\end{corollary}

As introduced in \cite{WA2013pacd}, we can also consider a kind of bisimulations between genuine models and pre-models.

\begin{definition}[Trans-bisimulation \cite{WA2013pacd}]\label{def:trans-bis}
Let a model $\mfm=(M,\sim,V)$ and a pre-model $\mfn=(N,\backsim,\nu)$ be given. A non-empty binary relation $Z\subseteq M\times N$ is called a \emph{trans-bisimulation between $\mfm$ and $\mfn$}, if for all $m\in M$ and $n\in N$ with $mZn$:
\begin{description}
\item[(\at)]\ \ $m\in V(p)$ iff $n\in\nu(p)$ for all $p\in\prop$,
\item[(\zigag)]\ \ For all $m'\in M$ and all $i\in\ag$, if $m\sim_im'$ (and so $m\sim_{\{i\}}m'$), then there is an $n'\in N$ such that $m'Zn'$ and $n\backsim_{\tau_0}\cdots\backsim_{\tau_x}n'$ with each of $\tau_0, \ldots, \tau_x$ being ``$i$'' or ``$G$'' such that $i\in G$;
\item[(\ziggr)]\ \ For all $m'\in M$ and all $G\in\gr$ with $|G|\geq2$, if $m\sim_{G}m'$, then there is an $n'\in N$ such that $m'Zn'$ and $n\backsim_{G_1}\cdots\backsim_{G_x}n'$ with $G\subseteq G_1\cap\cdots\cap G_x$;
\item[(zag)]\ \ For all $n'\in N$ and all $\tau\in\ag\cup\gr$, if $n\backsim_\tau n'$, then there is an $m'\in M$ such that $m'Zn'$ and $m\sim_\tau m'$.
\end{description}
We write $Z:(\mfm,m)\tbis(\mfn,n)$ if $Z$ is a trans-bisimulation between $\mfm$ and $\mfn$ linking $m$ and $n$. We say a pointed model $(\mfm,m)$ and a pointed pre-model $(\mfn,n)$ are \emph{trans-bisimilar}, denoted by $(\mfm,m)\tbis(\mfn,n)$, if there is a trans-bisimulation $Z$ such that $Z:(\mfm,m)\tbis(\mfn,n)$.

To make the notation symmetric, we call $Z$ a \emph{trans-bisimulation between $\mfn$ and $\mfm$} if it is a trans-bisimulation between $\mfm$ and $\mfn$, and we regard $Z:(\mfn,n)\tbis(\mfm,m)$ just as $Z:(\mfm,m)\tbis(\mfn,n)$.
\end{definition}

(Pseudo) satisfaction of \rcd formulas is invariant under
trans-bisimulation. We will not prove that directly at this point: it
follows from a stronger result we prove later (Lemma
\ref{lemma:c-inv-tbis}).

\begin{theorem}[Pseudo soundness]\label{pseudo-soundness}
All theorems of \RDC\ are valid in the class of all pseudo models.
\end{theorem}
\begin{proof}
See the appendix.
\end{proof}

\paragraph{Finitary Canonical Models}

\begin{definition}[Closure]\label{def-closure}
Given a formula $\phi$, the \emph{closure} of $\phi$ is given by the function $cl:\globallangc\to\wp(\globallangc)$ which is defined as follows:
\begin{enumerate}[leftmargin=1.5em]
\item $\phi\in cl(\phi)$, and if $\psi\in cl(\phi)$, so are all of its subformulas;
\item \label{closure-neg}If $\phi$ is not a negation, then $\phi\in cl(\phi)$ implies $\neg\phi\in cl(\phi)$;
\item \label{closure-kidi} $K_i\psi\in cl(\phi)$ iff $D_{\{i\}}\psi\in cl(\phi)$;
\item \label{closure-cb}$C_G\psi\in cl(\phi)$ implies $\{K_iC_G\psi\ |\ a\in A\}\subseteq cl(\phi)$;
\item \label{closure-rneg} $\RGpath \neg\psi\in cl(\phi)$ implies $\RGpath \psi\in cl(\phi)$;
\item $\RGpath(\psi\wedge\chi)\in cl(\phi)$ implies $\{\RGpath \psi,\RGpath\chi\}\subseteq cl(\phi)$;
\item \label{closure-rkd}$\RGpath K_i\psi\in cl(\phi)$ implies $D_{\delta(\{i\},\grouppath)} \RGpath \psi\in cl(\phi)$;
\item $\RGpath D_H\psi\in cl(\phi)$ implies $D_{\delta(H,\grouppath)} \RGpath \psi\in cl(\phi)$;
\item $\RGpath C_H\psi\in cl(\phi)$ implies all of the following:
	\begin{itemize}[leftmargin=1em]
	\item $D_{\delta(H,\grouppath)}\RGpath C_H\psi\in cl(\phi)$, 
	\item $\{D_{\delta(\{i\},\grouppath)}\RGpath C_H\psi\ |\ i\in H\}\subseteq cl(\phi)$,
	\item $\RGpath\psi\in cl(\phi)$.
	\end{itemize}
\end{enumerate}
It is not hard to verify that the closure of a formula is finite.
\end{definition}

We use $\underline{\Gamma}$ as shorthand for $\bigwedge_{\phi\in\Gamma}\phi$ when $\Gamma$ is a finite set of formulas.

\begin{definition}[Canonical pseudo model]\label{canonical-model}
Let $\alpha$ be a formula. The \emph{canonical pseudo model} $\mfm^c=(S,\backsim,V)$ for $cl(\alpha)$ is defined below:
\begin{itemize}
\item $S=\{\Gamma\ |\ \Gamma\text{ is maximal consistent in }cl(\alpha)\}$;
\item $\Gamma\backsim_i\Delta$ iff $\{K_i\phi\ |\ K_i\phi\in\Gamma\}=\{K_i\phi\ |\ K_i\phi\in\Delta\}$;
\item $\Gamma\backsim_G\Delta$ iff $\{D_H\phi\ |\ D_H\phi\in\Gamma\}=\{D_H\phi\ |\ D_H\phi\in\Delta\}$ whenever $H\subseteq G$;
\item $V(p)=\{\Gamma\in S\ |\ p\in\Gamma\}$.
\end{itemize}
\end{definition}

\begin{proposition}\label{prop:can-ps}
The canonical pseudo model for any $cl(\alpha)$ is a pseudo model.
\end{proposition}
\begin{proof}
See the appendix.
\end{proof}

\begin{lemma}\label{lemma0}
Let $\mcs=\{\Gamma\ |\ \Gamma\text{ is maximal consistent in }cl(\alpha)\}$ with $\alpha$ a formula. It holds that $\vdash\bigvee_{\Gamma\in\mcs}\underline{\Gamma}$
and $\vdash\phi\lra\bigvee_{\phi\in\Gamma\in\mcs}\underline{\Gamma}$\, for all $\phi\in cl(\alpha)$.
\end{lemma}
\begin{proof}
See \cite[Exercise 7.16]{vanDitmarsch07del} for the first result
(although $cl(\alpha)$ is different in our case the proof is exactly
the same). We give a proof of the second result in the appendix.
\end{proof}

Let $(S,\backsim\RELpath,V)$ be an update of a canonical pseudo model, and $\mfp=\ab{\Phi_0\asymp_{\tau_0}\cdots\asymp_{\tau_{n-1}}\Phi_n}$ where $\asymp$ stands for $\backsim\RELpath$ and every $\tau_x$ is an agent or a group. If all agents in $\tau_0,\ldots,\tau_{n-1}$ appears in $H$, we call $\mfp$ a \emph{$\ab{G_1\cdots G_n}$-resolved $H$-path (from $\Phi_0$)}; if a formula $\phi$ is such that $\phi\in\Phi_i$ for all $0\leq i\leq n$, we call $\mfp$ a \emph{canonical $\phi$-path}. 

\begin{lemma}\label{lemma1}
If $\Gamma$ and $\Delta$ are maximal consistent in $cl(\alpha)$, then
\begin{enumerate}[leftmargin=1.5em]
\item\label{item-dc} $\Gamma $ is \emph{deductively closed} in $cl(\alpha)$, i.e., $\Gamma\vdash\phi\Lra\phi\in\Gamma$ for any $\phi\in cl(\alpha)$;

\item\label{item-neg} If $\neg\phi\in cl(\alpha)$, then $\phi\in\Gamma\Lra\neg\phi\notin\Gamma$;

\item\label{item-wedge} If $\phi\wedge\psi\in cl(\alpha)$, then $\phi\wedge\psi\in\Gamma\Lra\phi\in\Gamma\ \&\ \psi\in\Gamma$;

\item\label{itemtt} If $\underline{\Gamma}\wedge\hat K_i\underline{\Delta}$ is consistent, $\Gamma\backsim_i\Delta$; if $\underline{\Gamma}\wedge\hat D_G\underline{\Delta}$ is consistent, $\Gamma\backsim_G\Delta$;

\item If $K_i\phi\in cl(\alpha)$, then $K_i\Gamma\vdash\phi\Lra K_i\Gamma\vdash K_i\phi$;
\item If $D_G\phi\in cl(\alpha)$, then $D_G\Gamma\vdash\phi\Lra D_G\Gamma\vdash D_G\phi$;

\item\label{itemC} If $C_G\phi\in cl(\alpha)$, then $C_G\phi\in\Gamma\Lra\forall\Delta(\Gamma\mathbin{\backsim_{C_G}}\Delta\Ra\phi\in\Delta)$;
\item\label{itemBC} If $\RGpath C_H\phi\in cl(\alpha)$, then $\RGpath C_H\phi\in\Gamma$ iff every $\ab{G_1\cdots G_n}$-resolved $H$-path from $\Gamma$ is a canonical $\RGpath\phi$-path.

\end{enumerate}
\end{lemma}
\begin{proof}
We give the proof of the clause \ref{itemBC} in the appendix. Other clauses are the same as in \cite[Lemma 49]{WA2013pacd} which can be traced back to \cite[Chapter 7]{vanDitmarsch07del}.
\end{proof}

\begin{lemma}[Pseudo truth]\label{truth-lemma}
Let $\mfm^c=(S,\backsim,V)$ be the canonical pseudo model for $cl(\alpha)$. For all groups \grouppath, all $\Gamma\in S$,  and all $\RGpath\phi\in cl(\alpha)$, it holds that
$$\RGpath \phi\in\Gamma\quad\iff\quad\mfm^c\MUpath,\Gamma\models\phi.$$
\end{lemma}
\begin{proof}
We show this lemma by induction on $\phi$.
\begin{itemize}[leftmargin=1.36em]
\item The base case. $\RGpath p\in\Gamma$ iff $p\in\Gamma$ (Proposition \ref{prop:reduction}(\ref{red-base})) iff $\mfm^c,\Gamma\models p$ iff $\mfm^c,\Gamma\models \RGpath p$ iff $\mfm^c\MUpath,\Gamma\models p$.

\item The case for negation. $\RGpath \neg\psi\in\Gamma$
iff $\neg \RGpath \psi\in\Gamma$ (note that $\neg \RGpath \psi\in cl(\alpha)$ by Definition \ref{def-closure}(\ref{closure-neg},\ref{closure-rneg}))
iff $\RGpath \psi\notin\Gamma$
iff $\mfm^c\MUpath,\Gamma\not\models\psi$
iff $\mfm^c\MUpath,\Gamma\models\neg\psi$.

\item The case for conjunction. $\RGpath (\psi\wedge\chi)\in\Gamma$
iff $(\RGpath \psi\wedge \RGpath \chi)\in\Gamma$\\
iff $\{\RGpath \psi,\RGpath \chi\}\subseteq\Gamma$ ($\RGpath \psi$ and $\RGpath \chi$ are in $cl(\alpha)$)\\
iff $\mfm^c\MUpath,\Gamma\models\psi$ and $\mfm^c\MUpath,\Gamma\models\chi$
iff $\mfm^c\MUpath,\Gamma\models\psi\wedge\chi$.

\item The case for individual knowledge. From left to right.\\
	\begin{tabular}{@{\ }l@{\ \ }l@{\ \ }l@{}}
	 &$\RGpath K_i\psi\in\Gamma$\\
	 iff&\multicolumn{2}{@{}l}{$D_\delta \RGpath \psi\in\Gamma$ where $\delta=\delta(\{i\},\grouppath)$}\\
	 iff&$\forall\Delta.(\Gamma\backsim_\delta\Delta\Ra D_\delta \RGpath \psi\in\Delta)$ &\\
	 $\Ra$&$\forall\Delta.(\Gamma\backsim_\delta\Delta\Ra \RGpath \psi\in\Delta)$ &(T$_D$)\\
	 iff&$\forall\Delta.(\Gamma\backsim_\delta\Delta\Ra\mfm^c\MUpath,\Delta\models\psi)$ &(IH)\\
	 iff&$\forall\Delta.(\Gamma\backsim_\delta\Delta\Ra\mfm^c,\Delta\models \RGpath \psi)$ \\
	 iff&$\mfm^c,\Gamma\models D_\delta \RGpath \psi$\\
	 iff&$\mfm^c,\Gamma\models \RGpath K_i\psi$&(\ref{cor-rd-red}(\ref{item-k-red}), \ref{prop:res-red-pseudo})\\
	 iff&$\mfm^c\MUpath,\Gamma\models K_i\psi$. 
	 \end{tabular}

From right to left. Suppose $\mfm^c\MUpath,\Gamma\models K_i\psi$. We must show $\RGpath K_i\psi\in\Gamma$. Suppose this is not the case. Then $\neg\RGpath K_i\psi\in\Gamma$. Hence $\underline\Gamma\wedge\neg\RGpath K_i\psi$ is consistent, and so is $\underline\Gamma\wedge\hat D_\delta\neg\RGpath\psi$, where $\delta=\delta(\{i\},\grouppath)$. Let $\mcs$ be the set of all maximal consistent sets in $cl(\alpha)$. By Lemma \ref{lemma0}, $\underline\Gamma\wedge\hat D_\delta\bigvee_{\neg\RGpath\psi\in\Theta\in\mcs}\underline{\Theta}$ is consistent. Since conjunction, resolution and the $\hat D_\delta$-operator all distribute over disjunction,  $\bigvee_{\neg\RGpath\psi\in\Theta\in\mcs}(\underline\Gamma\wedge\hat D_\delta\underline{\Theta})$ is consistent. Therefore there must be a $\Theta\in\mcs$ such that $\neg\RGpath\psi\in\Theta$ and $\underline\Gamma\wedge\hat D_\delta\underline{\Theta}$ is consistent.

\quad From $\neg\RGpath\psi\in\Theta$ we get $\RGpath\psi\notin\Theta$. By the induction hypothesis  $\mfm^c\MUpath,\Theta\not\models\psi$, and so $\mfm^c,\Theta\not\models\RGpath\psi$. By Lemma \ref{lemma1}(\ref{itemtt}) and that $\underline\Gamma\wedge\hat D_\delta\underline{\Theta}$ is consistent, $\Gamma\backsim_\delta\Theta$. But this contradicts the supposition that $\mfm^c\MUpath,\Gamma\models K_i\psi$, since by the same reasoning as in the proof of the other direction (see above),  $\mfm^c,\Delta\models \RGpath \psi$ for all $\Delta$ such that $\Gamma\backsim_\delta\Delta$.

\item The case for distributed knowledge: similar to the case for individual knowledge, and in the proof we use $\delta(H,\grouppath)$ instead of $\delta(\{i\},\grouppath)$.

\item The case for common knowledge.
$\RGpath C_H\psi\in\Gamma$\\
iff all $\ab{G_1\cdots G_n}$-resolved $H$-paths from $\Gamma$ are also canonical $\RGpath\psi$-paths.\\
Namely, for all $\Delta$ such that $(\Gamma,\Delta)\in(\backsim\RELpath)_{C_H}$, $\RGpath\psi\in\Delta$\\
iff for all $\Delta$ such that $(\Gamma,\Delta)\in(\backsim\RELpath)_{C_H}$, $\mfm^c\MUpath,\Delta\models\psi$ (by IH)\\
iff $\mfm^c\MUpath,\Gamma\models C_H\psi$.

\item The case for $R_H\psi$.
$\RGpath R_H\psi\in\Gamma$
iff $\mfm^c\MUpath|_H,\Gamma\models\psi$ (IH applies to $\psi$)\\
iff $\mfm^c\MUpath,\Gamma\models R_H\psi$.
\end{itemize}
\end{proof}

\begin{corollary}
Let $\mfm^c=(S,\backsim,V)$ be the canonical pseudo model for $cl(\alpha)$. For all $\Gamma\in S$ and all $\phi\in cl(\alpha)$, it holds that $\phi\in\Gamma$ iff $\mfm^c,\Gamma\models\phi.$
\end{corollary}

\begin{lemma}[Pseudo completeness]
Let $\phi$ be an $\globallangc$-formula. 
If $\phi$ is valid on all pseudo models, then it is provable in $\RCD$.
\end{lemma}

\paragraph{From Pseudo Completeness to Completeness}
\label{trans-proof}

By using \emph{unraveling} and
\emph{folding} from \cite[pp.\,9--15]{WA2013pacd}, we can transform
the canonical pseudo model to a bisimilar pre-model and then to a
trans-bisimilar proper model. It remains to show that this process
preserves truth. We will use $\tbis$ to denote the trans-bisimulation relation.

\begin{lemma}[Invariance of trans-bisimulation]
\label{lemma:c-inv-tbis}
\ \ Let $(\mfm,m)$ be a pointed model, $(\mfn,n)$ a pointed pre-model, and $(\mfs,s)$ a pointed pseudo model. If
$(\mfm,m)\tbis(\mfn,n)\bis(\mfs,s)$, then $\mfm,m\models\phi$ iff $\mfn,n\modelsp\phi$ for all formulas $\phi$.
\end{lemma}
\begin{proof}
The lemma can be shown by induction on $\phi$. Here we only show the
case for the resolution operators,  proofs of other cases are exactly as in the proof of \cite[Lemma 26]{WA2013pacd}.

Given a pointed model $(\mfm,m)$, a pointed pre-model $(\mfn,n)$ and a pointed pseudo model $(\mfs,s)$, such that $Z:(\mfm,m)\tbis(\mfn,n)$ for some $Z$ and $(\mfn,n)\bis(\mfs,s)$, we have the following:
$$\begin{array}{rcll}
\mfm,m\models R_G\psi&\iff&\mfm|_G,m\models\psi&\\
&\iff&\mfn|_G,n\modelsp\psi&(*)\\
&\iff&\mfn,n\modelsp R_G\psi,\\
\end{array}$$
where to show $(*)$ it is sufficient to show that $Z:(\mfm|_G,m)\tbis(\mfn|_G,n)$, as $(*)$ is then guaranteed by the induction hypothesis (note that $(\mfn|_G,n)\bis(\mfs|_G,s)$ by Proposition \ref{prop:bis-inv}).
Let $\mfm=(M,\sim,V)$ and $\mfn=(N,\backsim,\nu)$.
\begin{itemize}[leftmargin=1.2em]
\item The case for (\at) holds by $Z:(\mfm,m)\tbis(\mfn,n)$.

\item As for (\ziggr), suppose $m(\sim\rest_G)_{H}m'$ for some $m'\in M$ and $|H|\geq 2$.
	\begin{itemize}[leftmargin=.6em]
	\item If $G\cap H=\emptyset$, $(\sim\rest_G)_H=\sim_H$. By $Z:(\mfm,m)\tbis(\mfn,n)$ there is an $n'\in N$ such that $m'Zn'$ and $n\backsim_{H_0}\cdots\backsim_{H_x}n'$ with $H\subseteq H_0\cap\cdots\cap H_x$. Let $H_0=\cdots =H_x=H$. Thus $n\backsim_{H}\cdots\backsim_{H}n'$. Since $(\backsim\rest_G)_H=\backsim_H$, it holds that $n(\backsim\rest_G)_H\cdots(\backsim\rest_G)_Hn'$,  and so (\ziggr) holds in this case.
	
	\item If $G\cap H\neq\emptyset$, $(\sim\rest_G)_H=\sim_{G\cup H}$. By $Z:(\mfm,m)\tbis(\mfn,n)$ there is an $n'\in N$ such that $m'Zn'$ and $n\backsim_{H_0}\cdots\backsim_{H_x}n'$ with $G\cup H\subseteq H_0\cap\cdots\cap H_x$. Thus $n\backsim_{G\cup H}\cdots\backsim_{G\cup H}n'$.  Since $(\backsim\rest_G)_H=\backsim_{G\cup H}$,  It holds that $n(\backsim\rest_G)_H\cdots(\backsim\rest_G)_Hn'$. (\ziggr) holds also in this case.
	\end{itemize}
	
\item The case for (\zigag) is analogous.

\item The case for (zag). For all $n'\in N$ and all $\tau\in\ag\cup\gr$, if $n(\backsim\rest_G)_\tau n'$, then we must show that there is an $m'\in M$ such that $m'Zn'$ and $m(\sim\rest_G)_\tau m'$.
	\begin{itemize}[leftmargin=.6em]
	\item If $\tau$ is an agent $i$. Then if $i\in G$, $(\backsim\rest_G)_i=\backsim_G$, otherwise $(\backsim\rest_G)_i=\backsim_i$.  By $Z:(\mfm,m)\tbis(\mfn,n)$, we have $m\sim_Gm'$ if $i\in G$, or $m\sim_im'$ otherwise. Namely $m(\sim\rest_G)_i m'$ in either case.
	\item If $\tau$ is a group $H$. Then if $G\cap H=\emptyset$,  $(\backsim\rest_G)_H=\backsim_H$, otherwise $(\backsim\rest_G)_H=\backsim_{G\cup H}$.  By $Z:(\mfm,m)\tbis(\mfn,n)$, we have $m\sim_Hm'$ if $G\cap H=\emptyset$, or $m\sim_{G\cup H}m'$ otherwise. Namely $m(\sim\rest_G)_H m'$ in either case.
	\end{itemize}
\end{itemize}
We have shown that the lemma holds for the case for resolution. For other cases we refer to the proof of \cite[Lemma 26]{WA2013pacd}.
\end{proof}

\begin{theorem}[Completeness]
  For any \rcd formula $\phi$, if $\phi$ is valid then it is provable
  in \RCD.
\end{theorem}
\begin{proof}
  It suffices to show that any $\RCD$-consistent formula is
  satisfiable. Let $\phi$ be consistent. Let
  $\mfm^c$ be the canonical pseudo model for $cl(\phi)$. By the
  pseudo truth lemma (with $n=0$, i.e., an empty list of resolution operators), $\phi$ is satisfied in a state
  $\Gamma$ in $\mfm^c$.  Now let $\mfn^{\mfm^c}$ be the \emph{unraveling}
  \cite[Definition 18]{WA2013pacd}\footnote{While unraveling is a standard general
    technique; here we mean unraveling exactly in the sense of the
    mentioned definition.} of $\mfm^c$.  $\mfn^{\mfm^c}$ is a pre-model
  \cite[Proposition 19]{WA2013pacd}.  Now let $(\mfn^{\mfm^c})^*$ be the \emph{folding} \cite[Definition 22]{WA2013pacd} of $\mfn^{\mfm^c}$.  $(\mfn^{\mfm^c})^*$ is a (proper) model
  \cite[Definition 22]{WA2013pacd}.  From \cite[Lemma 27]{WA2013pacd} and
  \cite[Lemma 28]{WA2013pacd} we have that unraveling preserves
  bisimulation and that folding preserves trans-bisimulation, in other
  words we have that\footnote{Here $\overline{\Gamma}$ is any path in
    the unraveling starting with $\Gamma$.} $(\mfm,\Gamma) \bis
  (\mfn^{\mfm^c},\overline{\Gamma}) \tbis
  ((\mfn^{\mfm^c})^*,\overline{\Gamma})$.  By
  Corollary \ref{cor:premodelinvariance},
  $(\mfn^{\mfm^c},\overline{\Gamma}) \modelsp \phi$.  By Lemma
  \ref{lemma:c-inv-tbis}, $((\mfn^{\mfm^c})^*,\overline{\Gamma}) \models
  \phi$ and we are done.
\end{proof}

\section{Discussion}
\label{sec:conclusions}

In this paper we captured the dynamics of publicly observable private
resolution of distributed knowledge. Resolution operators (using
update semantics) are both an alternative and a complement to the
standard distributed knowledge operators (which use standard modal
semantics).

Resolution operators let us reason about the relationship between
common knowledge and distributed knowledge in general, and in particular about
distributed knowledge as potential common knowledge -- when can distributed knowledge become common knowledge?
A naive idea would be that $D_G \phi$ should imply that $R_GC_G\phi$
-- any information that is distributed can become common
knowledge through resolution. This does not hold in general, however,
due to Moore-like phenomena -- $\phi$ might even become false after
resolution (an example is the formula $D_{\{1,2\}}(p \wedge \neg
K_1p)$ discussed in the introduction).  We do, however, have the
following (Prop. \ref{prop:red-ck}(\ref{RC-subset}) with $G=H$):
\[R_GC_G \phi \leftrightarrow D_GR_G\phi.\] A fact can become common
knowledge after the group have shared their information if and only if
it was distributed knowledge before the event that the fact would be
true after the event. This is exactly the distributed knowledge that
can become common knowledge (in our special case of publicly
observable private resolution of distributed knowledge). If the grand
coalition resolves its distributed knowledge, there is no distinction
between distributed and common knowledge any more:
$R_{\ag}C_{\ag}\phi\lra R_{\ag}D_{\ag}\phi$
(Prop. \ref{red-grand-coalition}).

As discussed in the introduction, it has been argued that distributed
knowledge in general does not comply with the following \emph{principle
  of full communication} \cite{van1999group}: if $D_G \phi$ is true,
then $\phi$ follows logically from the set of all formulas known by at
least one agent in the group. This is seen as a problem: namely that
agents can have distributed knowledge without being able to establish
it ``through communication'' \cite{van1999group}.  Several papers
\cite{van1999group,gerbrandy:1999,roelofsen06dk} have tried to
characterize classes of models on which the principle of full
communication \emph{does} hold -- the class of all such models is
called \emph{full communication models} \cite{roelofsen06dk}.  This may
seem related to the distinction between distributed knowledge and
resolution operators: the latter is intuitively related to internal
``full'' communication in the group. However, this similarity is
superficial: the notion of full communication in the sense of
\cite{van1999group} is about \emph{expressive power of the
  communication language} and the limits that puts on the resulting
possible epistemic states under certain assumptions about how
information is shared. The key point of the resolution operators, on
the other hand, compared with the standard distributed knowledge
operators, is to make a distinction between \emph{before} and
\emph{after} the information sharing event. That distinction is not made in
standard distributed knowledge -- even restricted to full communication
models: it is easy to see that, e.g., $D_{\{1,2\}} (p \wedge \neg K_1
p)$ is satisfiable also on full communication models. The two ideas,
of limiting models to full communication models and of modeling group
information sharing events using model updates, are orthogonal, and there is
nothing against restricting logics with the resolution operators to
full communication models.  We leave that for future work.
Furthermore, it would be interesting to look at a combined variant:
``update by full communication'', which takes the communication
language into account when defining the updated model.

A main interest for future work is expressive power. Can it
be shown that $\globallangc$ is strictly more expressive than
$\langdc$?  Another, related, natural question is the relative
expressivity of $\rcd$ and $\pacd$: can the combination of public
announcement operators (which eliminate states) and distributed
knowledge operators (which pick out states considered possible by
everyone) always be used to ``simulate'' the resolution operators
(which eliminate states considered possible by everyone)?

Also of interest for future work is to look at other assumptions about
the other agents' knowledge about the group communication event taking
place. In this paper we only studied the case that it is common
knowledge that the event takes place (but not what the agents in the group
learn). That was naturally modeled using a ``global'' model update:
in every state, replace accessibility for each agent in the group with
the group accessibility (intersection). An interesting and also
natural alternative is doing only a ``local'' model update: change
accessibility in the same way, but only in the current state. That
would correspond to it being common knowledge that if this is the
current state, then the group resolves their knowledge.

When looking at the interaction of the resolution and common knowledge
operators one might be reminded of \emph{relativized common knowledge}
\cite{vb1999relativization,vb06logics}.  Here is an open question: can
$R_GC_H\phi$ be expressed using relativized common knowledge, in
combination with other operators?

Finally, there is a conceptual relationship to \emph{group
  announcement logic} \cite{jal}, where formulas of the form $\langle
G\rangle\phi$ say that $G$ can make a joint public announcement such
that $\phi$ will become true. A difference to the resolution operators
in this paper is that latter model private communication. Yet, the
exact relationship between these operators is interesting for future
work.

\section{Acknowledgments}

Y{\`i} N. W\'ang acknowledges funding support from the Scientific Research Foundation for the Returned Overseas Chinese Scholars, State Education Ministry of P.R.C.

\bibliographystyle{eptcs}
\bibliography{resolving}

\appendix

\section{Some proofs}
\paragraph{Proof of Proposition \ref{prop:red-ck}}
\ref{RC-no-intersection}.
  $\mfm,s \models \rg{G}C_H \phi$ iff $\mfm|_G,s \models C_H\phi$ iff $\mfm|_G,t \models \phi$ for any $(s,t) \in \sim^{*'}_H$, where $\sim^{*'}_H = (\bigcup_{i \in H} \sim_i')^*$ and $\sim_i' = \bigcap_{j \in G} \sim j$ for $i \in G$ and $\sim_i' = \sim_i$ for $i \not\in G$. Thus, when $G \cap H = \emptyset$, we get that $\sim^{*'}_H = (\bigcup_{i \in H} \sim_i)^*$. $\mfm|_G,t \models \phi$ for any $(s,t) \in (\bigcup_{i \in H} \sim_i)^*$ holds iff $\mfm,t \models \rg{G}\phi$ for any $(s,t) \in (\bigcup_{i \in H} \sim_i)^*$ iff $\mfm,t \models C_H\rg{G}\phi$.
  
\ref{RC-subset}.
$$\begin{array}{ll}
& \mfm,s\models R_GC_H\phi\\
\iff&	\mfm|_G,s\models C_H\phi\\
\iff&	\mfm|_G,t\models \phi \text{ for all $t$ s.t. }(s,t)\in\sim^G_{C_H}\\[1ex]
\iff\ ^\dag&	\mfm|_G,t\models \phi \text{ for all $t$ s.t. }(s,t)\in\sim^G_{i}\\
\iff&	\mfm|_G,s\models K_i\phi\\
\iff&	\mfm,s\models R_GK_i\phi\\
\end{array}$$

For the $\dag$ step, note that when $i\in G$, $\sim^G_i=\sim^G_j$ for any $j\in G$ (and actually also equal to $\sim_G$). Therefore,
$$\sim^G_{C_H}=(\bigcup_{i\in H}\sim^G_i)^*=(\sim^G_i)^*=\sim^G_i.$$

That $R_GK_i\phi\lra D_GR_G\phi$ is valid is already shown in Proposition \ref{prop:reduction}.

\paragraph{Proof of Proposition \ref{prop:up-pres-ps}}
Let $\mfm=\premodelEx$. Clearly $\mfm|_G=(S,\backsim\rest_G,V)$ is a pre-model. Moreover,
\begin{enumerate}
\item Given an agent $i$, 
$$\begin{array}{lll}
(\backsim\rest_G)_{\{i\}}&=&
	\left\{\begin{array}{ll}
	\backsim_{\{i\}\cup G}&i\in G\\
	\backsim_{\{i\}},& i\notin G
	\end{array}\right.\\[1em]
&=&\left\{\begin{array}{ll}
\backsim_G&i\in G\\
\backsim_i,& i\notin G
\end{array}\right.\\[1em]
&=&(\backsim\rest_G)_i.
\end{array}$$

\item Given two groups $H$ and $H'$ such that $H\subseteq H'$,
$$\begin{array}{lll}
(\backsim\rest_G)_{H'}&=&
	\left\{\begin{array}{ll}
	\backsim_{H'\cup G}&H'\cap G\neq\emptyset\\
	\backsim_{H'},& H'\cap G=\emptyset\\
	\end{array}\right.\\[1em]
(\backsim\rest_G)_{H}&=&
	\left\{\begin{array}{ll}
	\backsim_{H\cup G}&H\cap G\neq\emptyset\\
	\backsim_{H},& H\cap G=\emptyset\\
	\end{array}\right.	
\end{array}$$
So we have:
	\begin{itemize}[leftmargin=.5em]
	\item when $H\cap G\neq\emptyset$ (and therefore $H'\cap G\neq\emptyset$), $(\backsim\rest_G)_{H'}\ =\ \backsim_{H'\cup G}\ \subseteq\ \backsim_{H\cup G}\ =\ (\backsim\rest_G)_H$;
	\item when $H'\cap G=\emptyset$ (and therefore $H\cap G=\emptyset$), $(\backsim\rest_G)_{H'}=\backsim_{H'}\subseteq\backsim_H=(\backsim\rest_G)_H$;
	\item otherwise $H'\cap G\neq\emptyset$ and $H\cap G=\emptyset$, and in this case $(\backsim\rest_G)_{H'}\ =\ \backsim_{H'\cup G}\ \subseteq\ \backsim_H\ =\ (\backsim\rest_G)_H$.
	\end{itemize}
\end{enumerate}
$\mfm|_G$ is a pre-model satisfying the two conditions above, which shows it is a pseudo model.

\paragraph{Proof of Proposition \ref{prop:bis-inv}}
Let $\mfm=(S,\backsim,V)$ and $\mfm'=(S',\backsim',V')$. Thus $\mfm|_G=(S,\backsim\rest_G,V)$ and $\mfm'|_G=(S',\backsim'\rest_G,V')$. Suppose $Z:(\mfm,s)\bis(\mfm',s')$, and we show $Z:(\mfm|_G,s)\bis(\mfm'|_G,s')$:
\begin{itemize}
\item[(\at)] This clearly follows from the (\at) clause of $Z:(\mfm,s)\bis(\mfm',s')$.
\item[(zig)] For all $t\in S$, if $s(\backsim\rest_G)_Ht$, then 
	\begin{itemize}[leftmargin=0em]
	\item If $G\cap H=\emptyset$, then $(\backsim\rest_G)_H=\backsim_H$ and $(\backsim'\rest_G)_H=\backsim'_H$. By $Z:(\mfm,s)\bis(\mfm',s')$ there must be a $t'\in S'$ such that $s'(\backsim'\rest_G)_Ht'$ and $tZt'$.
	\item If $G\cap H\neq\emptyset$, then $(\backsim\rest_G)_H=\backsim_{G\cup H}$ and $(\backsim'\rest_G)_H=\backsim'_{G\cup H}$. By $Z:(\mfm,s)\bis(\mfm',s')$ there must be a $t'\in S'$ such that $s'(\backsim'\rest_G)_Ht'$ and $tZt'$.
	\end{itemize}
If $s(\backsim\rest_G)_it$, we can prove analogously to the above.

\item[(zag)] This can be shown analogously to the case for (zig).
\end{itemize}

\paragraph{Proof of Theorem \ref{pseudo-soundness}}
It is easy to verify that (\SFIVE), (\CK), (\DK), (N$_R$), (RA), (RC) and (RN) are all valid or admissible with respect to the class of all pseudo models. Here we only show that i) (RD1) and (RD2) are valid in all pseudo models, and ii) (RR$_C$) preserves validity of pseudo models.

Let $\mfm=(S,\backsim,V)$ be a pseudo model and $s\in S$. We show the following:
\begin{itemize}[leftmargin=1.2em]
\item $\mfm,s\modelsp\text{RD1}$ and $\mfm,s\modelsp\text{RD2}$, i.e.,
	\begin{itemize}[leftmargin=1em]
	\item If $G \cap H \neq \emptyset$, then $\mfm,s\modelsp\rg{G}D_H\phi \leftrightarrow D_{G \cup H}\rg{G}\phi$;
	\item If $G \cap H = \emptyset$, then $\mfm,s\modelsp\rg{G}D_H\phi \leftrightarrow D_{H}\rg{G}\phi$.
	\end{itemize}
$$\begin{array}{@{}l@{\ }l@{}}
&\mfm,s\modelsp R_GD_H\phi\\
\iff&\mfm|_G,s\modelsp D_H\phi\\
\iff&\mfm|_G,t\modelsp \phi\text{ for all $t$ s.t. $(s,t)\in(\backsim\rest_G)_H$}\\
\iff&\mfm,t\modelsp R_G\phi\text{ for all $t$ s.t. $(s,t)\in(\backsim\rest_G)_H$}\\
\iff\,^\dag&\text{if }G\cap H\neq\emptyset,\ \mfm,t\modelsp R_G\phi\text{ for all $t$ s.t. $(s,t)\in\backsim\rest_{G\cup H}$,}\\
&\text{if }G\cap H=\emptyset,\ \mfm,t\modelsp R_G\phi\text{ for all $t$ s.t. $(s,t)\in\backsim\rest_{H}$}\\
\iff&\text{if }G\cap H\neq\emptyset,\ \mfm,t\modelsp D_{G\cup H}R_G\phi,\text{ and}\\
&\text{if }G\cap H=\emptyset,\ \mfm,t\modelsp D_HR_G\phi,\\
\end{array}$$
where the $^\dag$ step is by definition:
$$(\backsim\rest_G)_H=\left\{
\begin{array}{ll}
\backsim_{H\cup G},&G\cap H\neq\emptyset,\\
\backsim_H,&G\cap H=\emptyset.\\
\end{array}\right.$$

\item $\mfm,s\modelsp\phi\ra\RGpath C_H\psi$ under the assumption $\modelsp \phi\ra(E_H\phi\wedge\RGpath \psi)$. The proof is similar to the proof for genuine models.
\end{itemize}

\paragraph{Proof of Proposition \ref{prop:can-ps}}
Suppose that $\mfm=(S,\backsim,V)$ is the canonical pseudo model for $cl(\alpha)$.
We need to show that $\mfm$ is a pseudo model. Namely,
\begin{enumerate}
\item $S$ is non-empty, and
\item all $\backsim_i$'s and $\backsim_G$'s are equivalence relations, and
\item $V$ is a valuation from $\prop$ to $\wp(S)$, and
\item $\backsim_i=\backsim_{\{i\}}$ for every agent $i$, and
\item $\backsim_H\subseteq \backsim_G$ if $G$ and $H$ are groups such that $G\subseteq H$.
\end{enumerate}
Conditions 1--3 are the conditions for being a pre-model which are easy to verify. Conditions 4 and 5 are additional conditions for being a pseudo model.

By Definition \ref{def-closure}(\ref{closure-kidi}), $K_i\phi$ and $D_i\phi$ must be in $cl(\alpha)$ both or neither. Thus, for any $\Gamma,\Delta\in S$,
$$\begin{array}{lll}
&\Gamma \backsim_i\Delta\\
\iff&\{K_i\phi\ |\ K_i\phi\in\Gamma\}=\{K_i\phi\ |\ K_i\phi\in\Delta\}&\\
\iff&\{D_i\phi\ |\ D_i\phi\in\Gamma\}=\{D_i\phi\ |\ D_i\phi\in\Delta\}&\text{(Axiom DK1)}\\
\iff&\Gamma\backsim_{\{i\}}\Delta.&\\
\end{array}$$
$$\begin{array}{lll}
&\Gamma\backsim_H\Delta\\
\iff&\{D_{H'}\phi\ |\ D_{H'}\phi\in\Gamma\}=\{D_{H'}\phi\ |\ D_{H'}\phi\in\Delta\},\\
&\text{ for any group }H'\subseteq H\\
\Ra&\{D_{G'}\phi\ |\ D_{G'}\phi\in\Gamma\}=\{D_{G'}\phi\ |\ D_{G'}\phi\in\Delta\},\\
&\text{ for any group }G'\subseteq G\\
\iff&\Gamma\backsim_{G}\Delta.\\
\end{array}$$
This finishes the proof, and shows that the notion ``canonical pseudo model'' is well-defined.

\paragraph{Proof of Lemma \ref{lemma0}(2)}
Let $\phi\in cl(\alpha)$. By $\vdash(\bigvee_{\neg\phi\in\Gamma\in\mcs}\underline\Gamma)\ra\neg\phi$ and the first result of this lemma (i.e., $\vdash\bigvee_{\neg\phi\in\Gamma\in\mcs}\vee\bigvee_{\phi\in\Gamma\in\mcs}$) we get $\vdash\phi\ra\bigvee_{\phi\in\Gamma\in\mcs}\underline{\Gamma}$. For the converse direction, suppose $\nvdash\bigvee_{\phi\in\Gamma\in\mcs}\underline{\Gamma}\ra\phi$. Then $\neg(\bigvee_{\phi\in\Gamma\in\mcs}\underline{\Gamma}\ra\phi)$ is consistent. Namely $\neg\phi\wedge\bigvee_{\phi\in\Gamma\in\mcs}\underline{\Gamma}$ is consistent. But this is impossible.

\paragraph{Proof of Lemma \ref{lemma1}(\ref{itemBC})}
Let $\RGpath C_H\phi\in cl(\alpha)$. It follows from the definition of closure (Definition \ref{def-closure}) that the following formulas:
\begin{itemize}
\item $D_{\delta(\{i\},\grouppath)} \RGpath C_H\phi$ where $i\in H$
\item $D_{\delta(H,\grouppath)} \RGpath C_H\phi$
\item $\RGpath\phi$ and $\neg\RGpath\phi$
\end{itemize}
are all in $cl(\alpha)$.
 
From left to right. Suppose $\RGpath C_H\phi\in\Gamma$, we continue by induction on the length of the path that every $\ab{G_1\cdots G_n}$-resolved $H$-path from $\Gamma$ is a canonical $\RGpath C_H\phi$-path. Then the left-to-right direction follows: by $\vdash C_H\phi\ra\phi$, N$_R$ and $R_G$-distribution (which follows from \RR\ axioms) we get $\vdash \RGpath C_H\phi\ra \RGpath\phi$, and by $\RGpath\phi\in cl(\alpha)$ we have $\RGpath\phi\in\Gamma$.

Suppose the length of the $\ab{G_1\cdots G_n}$-resolved $H$-path is 0, i.e., the path is $\ab{\Gamma}$, we must show that $\RGpath C_H\phi\in\Gamma$. This is guaranteed by the supposition. 

Suppose the length of the $\ab{G_1\cdots G_n}$-resolved $H$-path is $n+1$, i.e., the path is $\ab{\Gamma_0\asymp_{\tau_0}\cdots\asymp_{\tau_{n-1}}\Gamma_n\asymp_{\tau_n}\Gamma_{n+1}}$ with $\Gamma_0=\Gamma$ and every $\tau_x$ is either in $H$ or a subset of $H$. By the induction hypothesis we may assume that $\RGpath C_H\phi\in\Gamma_n$.
	\begin{itemize}[leftmargin=.9em]
	\item Suppose $\tau_n$ is an agent $i$ ($i\in H$). By Axiom C1 we have $\vdash C_H\phi\ra K_iC_H\phi$. It follows that $\vdash \RGpath C_H\phi\ra \RGpath K_iC_H\phi$ by the rules N$_R$ and $R_G$-distribution. Let $\delta=\delta(\{i\},\grouppath)$. By the reduction axioms we move $K_i$ left, i.e., $\vdash \RGpath K_iC_H\phi\ra D_\delta \RGpath C_H\phi$, so we get $\vdash \RGpath C_H\phi\ra D_\delta \RGpath C_H\phi$. Hence $\Gamma_n\vdash D_\delta \RGpath C_H\phi$. As $D_\delta \RGpath C_H\phi\in cl(\alpha)$, we have $D_\delta \RGpath C_H\phi\in\Gamma_n$.
	Moreover, by Proposition \ref{prop:res-red-pseudo}, $\asymp_i=\backsim_\delta$. Thus $D_\delta \RGpath C_H\phi\in\Gamma_{n+1}$ by the definition of $\backsim_\delta$, and so $\RGpath C_H\phi\in\Gamma_{n+1}$.
	
	\item Suppose $\tau_n$ is a group $I$ ($I\subseteq H$). By Axioms C1, D1 and D2 we have $\vdash C_H\phi\ra D_IC_H\phi$. By N$_R$ and $R_G$-distribution, $\vdash \RGpath C_H\phi\ra \RGpath D_IC_H\phi$. By similar reasoning to the case above, we get the result $\RGpath C_H\phi\in\Gamma_{n+1}$ (we use $\delta(H,\grouppath)$ instead of $\delta(\{i\},\grouppath)$ in this case).
	\end{itemize}
In both cases we get $\RGpath C_H\phi\in\Gamma_{n+1}$ as we wish to show.

From right to left. Suppose that every $\ab{G_1\cdots G_n}$-resolved $H$-path from $\Gamma$ is a canonical $\RGpath\phi$-path. Let $\mcs_0$ be the set of all maximal consistent sets $\Delta$ in $cl(\alpha)$ such that every $\ab{G_1\cdots G_n}$-resolved $H$-path from $\Delta$ is a canonical $\RGpath\phi$-path. Now consider the formula
$$\lambda=\bigvee_{\Delta\in\mcs_0}\underline{\Delta}$$
We will show the following:
\begin{enumerate}
\item $\vdash\underline{\Gamma}\ra \lambda$
\item $\vdash \lambda\ra(E_H\lambda\wedge\RGpath\phi)$.
\end{enumerate}
From the above and the reduction rule for resolved common knowledge we get $\vdash\underline{\Gamma}\ra \RGpath C_H\phi$ which furthermore entails $\RGpath C_H\phi\in\Gamma$. We now continue with the proof of the two clauses.
\begin{enumerate}[wide]
\item This is trivial, as $\underline{\Gamma}$ is one of the disjuncts of $\lambda$.

\item Suppose towards a contradiction that $$\lambda\wedge\neg(E_H\lambda\wedge\RGpath\phi)$$ is consistent, i.e., $\lambda\wedge(\neg E_H\lambda\vee\neg\RGpath\phi)$ is consistent. Because $\lambda$ is a disjunction there must be a disjunct $\underline{\Xi}$ of $\lambda$ such that $\underline{\Xi}\wedge(\neg E_H\lambda\vee\neg\RGpath\phi)$ is consistent. It follows that either $\underline{\Xi}\wedge\neg E_H\lambda$ or $\underline{\Xi}\wedge\neg\RGpath\phi$ is consistent.

If the former is consistent, then there must be an agent $i\in H$ such that $\underline{\Xi}\wedge\neg K_i\lambda$ is consistent, i.e., $\underline{\Xi}\wedge\hat K_i\neg\bigvee_{\Delta\in\mcs_0}\underline{\Delta}$ is consistent. Since $\vdash\bigvee_{\Delta\in\mcs}\underline{\Delta}$ by Lemma \ref{lemma0}, we have $\vdash\neg\bigvee_{\Delta\in\mcs_0}\underline{\Delta}\ra\bigvee_{\Delta'\in\mcs\setminus\mcs_0}\underline{\Delta'}$, and so there must be a $\Theta$ in $\mcs\setminus\mcs_0$ such that $\underline{\Xi}\wedge\hat K_i\underline{\Theta}$ is consistent. By item \ref{itemtt} of this lemma $\Xi\backsim_i\Theta$ (where $\backsim$ is the relation in the canonical pseudo model for $cl(\alpha)$). But then $\Xi$ cannot be in $\mcs_0$ for $\Theta\notin\mcs_0$. A contradiction!

If the latter is consistent, since $\neg\RGpath\phi\in cl(\alpha)$ and $\Xi$ is maximal, $\neg\RGpath\phi\in\Xi$. But $\RGpath\phi\in\Xi$ since $\Xi\backsim_H\Xi$ and every $\ab{G_1\cdots G_n}$-resolved $H$-path from $\Xi$ is a canonical $\RGpath\phi$-path. We reach a contradiction.
\end{enumerate}

\end{document}